\documentclass[pdflatex,sn-mathphys-num]{sn-jnl}
\usepackage{graphicx}%
\usepackage{multirow}%
\usepackage{amsmath,amssymb,amsfonts}%
\usepackage{amsthm}%
\usepackage{mathrsfs}%
\usepackage[title]{appendix}%
\usepackage{xcolor}%
\usepackage{textcomp}%
\usepackage{manyfoot}%
\usepackage{booktabs}%
\usepackage{verbatim}
\usepackage{algorithm}%
\usepackage{algorithmicx}%
\usepackage{algpseudocode}%
\usepackage{listings}%
\newtheorem{lemma}{Lemma}
\newtheorem{cor}{Corollary}
\def\x{\boldsymbol{x}}

\def\0{{\mathbf 0}}

\newcommand{\F}{\mathbb{F}}
\newcommand{\Z}{\mathbb{Z}}

\newcommand{\Mod}[1]{\ (\mathrm{mod}\ #1)}

\theoremstyle{thmstyleone}%
\newtheorem{theorem}{Theorem}
\theoremstyle{thmstyletwo}%
\newtheorem{remark}{Remark}%

\theoremstyle{thmstylethree}%

\begin{document}

\title[Hermitian Hulls of Rational Algebraic Geometry Codes and Applications in Quantum Codes]{Hermitian Hulls of Rational Algebraic Geometry Codes and Applications in Quantum Codes}


\author*[1]{\fnm{Lin} \sur{Sok}}\email{lin.sok@ntu.edu.sg}

\author[1]{\fnm{Martianus Frederic} \sur{Ezerman}}\email{fredezerman@ntu.edu.sg}

\author[1,2]{\fnm{San} \sur{Ling}}\email{lingsan@ntu.edu.sg; ling.s@vinuni.edu.vn}

\affil*[1]{\orgdiv{School of Physical and Mathematical Sciences}, \orgname{Nanyang Technological University}, \orgaddress{\street{21 Nanyang Link}, \city{Singapore}, \postcode{637371}, \state{} \country{Singapore}}}

\affil[2]{\orgdiv{}\orgname{VinUniversity}, \orgaddress{\street{Vinhomes Ocean Park, Gia L\^{a}m}, \city{Hanoi}, \postcode{100000}, \state{}\country{Vietnam}}}


\abstract{Interest in the hulls of linear codes has been growing rapidly. More is known when the inner product is Euclidean than Hermitian. A shift to the latter is gaining traction. The focus is on a code whose Hermitian hull dimension and dual distance can be systematically determined. Such a code can serve as an ingredient in designing the parameters of entanglement-assisted quantum error-correcting codes (EAQECCs).

We use tools from algebraic function fields of one variable to efficiently determine a good lower bound on the Hermitian hull dimensions of generalized rational algebraic geometry (AG) codes. We identify families of AG codes whose hull dimensions can be well estimated by a lower bound. Given such a code, the idea is to select a set of evaluation points for which the residues of the Weil differential associated with the Hermitian dual code has an easily verifiable property.

The approach allows us to construct codes with designed Hermitian hull dimensions based on known results on Reed-Solomon codes and their generalization. Using the Hermitian method on these maximum distance separable (MDS) codes with designed hull dimensions yields two families of MDS EAQECCs. We confirm that the excellent parameters of the quantum codes from these families are new.}




\maketitle

\section{Introduction}\label{sec1}

Two algorithms spurred the quest to build quantum computers of a large-enough scale for cryptography. Shor in \cite{Shor} proposed a general quantum attack algorithm that can completely break the existing public key infrastructures. Grover in \cite{Grover} introduced a quantum algorithm with a roughly quadratic speed improvement over its classical counterpart on searching over any unsorted data. 

Quantum error-correcting codes, or {\em quantum codes} in short, form an essential component in controlling noise and decoherence. Qubit \emph{stabilizer codes}, which encode information on $2$-level quantum ensembles, were introduced by D. Gottesman in \cite{Gottesman97} and situated firmly in a general mathematical framework by Calderbank, Rains, Shor, and Sloane in \cite{CRSS98}. Generalization to qudit ($q$-level) quickly followed, with the work of Ketkar, Klappenecker, Kumar, and Sarvepalli in \cite{Ket Kla} being a good place to consult for further details. Stabilizer codes can be constructed via classical codes that satisfy certain orthogonality conditions, typically defined based on the Euclidean or Hermitian inner product. 

A different approach comes in the form of {\em entanglement-assisted quantum error correcting codes} (EAQECCs). These codes were first proposed by Bowen in \cite{Bow} and made popular by Brun, Devetak, and Hsieh in \cite{BDH06}. They showed that the orthogonality conditions can be relaxed, provided that the communicating parties can prepare and maintain a number of pre-shared pairs of entangled states. Wilde and Brun showed how to construct EAQECCs via classical codes in \cite{WilBru}.

The {\em hull} of a linear code is the intersection of the code with its dual code, where the dual is defined with respect to some inner product. The hull dimension is useful in deriving some parameters of an EAQECC. For results on the Euclidean hulls with application to EAQECCs, we can consult, {\it e.g.}, the works done in \cite{LCC,PerPel,Sok1D2}.

The quantum version of the Singleton bound for stabilizer codes as well as EAQCCs had been studied quite extensively. Quantum codes whose parameters meet the relevant Singleton-type bound with equality are said to be \emph{maximum distance separable} (MDS). New Singleton-like bounds for EAQECCs have recently been derived by Grassl, Huber, and Winter in \cite{GHW}, also to correct inaccuracies in prior versions of the bounds. Guenda, Gulliver, Jitman, and Thipworawimon studied the $\ell$-intersection pair of linear codes in \cite{GGJT18}. They determined, using the Euclidean construction, the parameters of \emph{all} qudit MDS EAQECCs of length $n \leq q+1$. 

For lengths $n > q+1$ one can use the Hermitian route. The resulting EAQECCs have better parameters than those from the Euclidean one. By studying the Hermitian hulls of MDS linear codes, Fang, Fu, Li, and Zhu in \cite{FangFuLiZhu} and, separately, Pereira, Pellikaan, La Guardia, and Assis in \cite{PerPel} built qudit MDS EAQECCs from generalized Reed-Solomon codes and one-point rational AG codes. Other works on qudit MDS EAQECCs include \cite{ChenHuangFengChen,FanChenXu,LSEL,WangZhuSun,Sari} and a good number of their references. Assuming the classical MDS conjecture, nontrivial MDS EAQECCs which are derived from classical codes have restrictive code lengths. For more detailed treatment on $q$-ary non-MDS EAQECCs when the values of $q$ are small, interested readers are referred to the works in \cite{SES,SESM,SQ}. Some propagation rules, which can also serve as tools for performance comparison among codes from different constructions, have been given in \cite{Anderson24,Luo2022}.

Algebraic geometry (AG) codes are known in the literature to be good classical ingredients in the constructions of both stabilizer codes and EAQECCs. The Euclidean dual of a given AG code is characterized in \cite{Stich}. The Euclidean hull has been recently explored in \cite{PerPel,Sok1D2}. Studying the Hermitian hull is more challenging since we know comparatively little on its characterization. 

The Hermitian hull of a one-point rational AG code for some special lengths has been treated in \cite{PerPel}. There, the dimension is computed by examining a basis. Entanglement is not a freely available resource. Creating, distributing, and maintaining entangled states incur costs. In the entanglement-assisted setup, one typically prefers codes that require smaller number of pre-shared entangled states. To design EAQECCs with good parameters, we want codes with large Hermitian hull dimensions. This motivates our study on the Hermitian hulls of one-point generalized rational AG codes. We devise $\mathbb{F}_{q^2}$-linear MDS codes whose Hermitian hull dimension can be lower bounded by a quantity close to the actual value.

Our contributions can be summarized into three insights.
\begin{enumerate}
\item Lemma \ref{lem:dim-hull} serves as a key to determine a good lower bound on the Hermitian hull dimension of a (generalized) rational AG code. The lemma leads to Theorem \ref{thm:2new} which constructs one-point AG codes whose Hermitian hull dimensions can be explicitly found. The idea is to select a set of evaluation points for which the residues of the Weil differential associated with the dual code are the $(q+1)^{\rm st}$ power elements in $\F_{q^2}$. 

\item Theorem \ref{thm:2new} allows us to construct codes with designed hull dimensions based on known results on Reed-Solomon codes and their generalization, depending on their sets of evaluation points, as explained in Corollaries \ref{cor:1} and \ref{cor:2}. To the best of our knowledge, not much has been done on the determination of the Hermitian hull dimension of an AG code. Pereira {\it et al.} in \cite{PerPel} considered the Hermitian hulls of one-point rational AG codes over $\F_{q^2}$ for maximal length, which is $q^2$. Here, we treat diverse lengths. 

\item Using the Hermitian method, we build qudit MDS EAQECCs with specific number of pre-shared entangled states from $q^2$-ary linear MDS codes. We subsequently obtain two families of MDS EAQECCs, formally stated in Theorems \ref{thm:Q0} and \ref{thm:Q1}. We confirm that the parameters of the codes from these families are new. Theorem \ref{thm:summary} summarizes sufficient conditions that ensure the existence of MDS EQAECCs of the specified parameters.
\end{enumerate}

\begin{theorem}\label{thm:summary}
Let $q$ be a prime power. Let $n_0$ and $q_1$ be integers such that $1 \le n_0\le q-1$ and $0\le q_1 \le q-1$. Let $k = k_0 \, q + q_0$, with $1 \le k_0 < \lfloor (n_0 \, q-q_0)/q \rfloor$, $0\le q_0\le q-1$, and $q_1-q_0\le 1$. If $\ell$ is defined, in cases, as
\begin{equation*}
\ell=
\begin{cases}
k_0(n_0-k_0)+q_0+1 \mbox{ if } k_0\le q_1+q-q_0-2 \mbox{ and } q_0 \le n_0-k_0-2,\\
(k_0+1)(n_0-k_0) \mbox{ if } k_0 \le q_1+q-q_0-2 \mbox{ and } q_0\ge n_0-k_0-1,\\
k_0(n_0-k_0-1)+q_1+q \mbox{ if }  k_0> q_1+q-q_0-2 \mbox{ and } q_0< n_0-k_0-2,\\
(n_0-k_0-1)(k_0+1)+(q_1+q-q_0-1) \mbox{ if } k_0> q_1+ q-q_0-2 \mbox{ and } q_0 \ge n_0-k_0-2,
\end{cases}
\end{equation*}
then the following assertions hold.
\begin{enumerate}
\item There are quantum codes ${\cal Q}_1$ and ${\cal Q}_2$ with respective parameters
\begin{align*}
&[[n_0 \, q, k+1-\ell,  n_0 \, q-(k+1) ; n_0 \, q - (k+1)-\ell]]_{q} \mbox{ and}\\
&[[n_0 \, q, n_0 \, q-(k+1) -\ell,  (k+1);(k+1)-\ell]]_{q}.
\end{align*}
\item Let $t$, $s$, and $r$ be integers such that $s$ divides $(q^{2}-1)$, $\displaystyle{r=\frac{s}{\gcd (s,q+1)}}$ and $1 \leq t < \frac{q-1}{r}$. For $n = (t+1) \, s + 1 =n_0 \, q + q_1$, there are quantum codes ${\cal Q}_1$ and ${\cal Q}_2$ with respective parameters
\begin{align*}
&[[n, k+1-\ell,  n-(k+1);n-(k+1)-\ell]]_{q} \mbox{ and}\\
& [[n, n-(k+1)-\ell,  (k+1);(k+1)-\ell]]_{q}.
\end{align*}
\end{enumerate}
\end{theorem}

After this introduction, Section \ref{sec:pre} gathers basic notions, definitions, and useful known results on function fields, algebraic geometry codes, and related codes. Section \ref{sec:main} deals with the Hermitian hulls of one-point rational AG codes. The focus is on codes whose hull dimensions can be determined by using the bases of the codes and their dual. We then provide a formula to lower bound the hull dimensions. Section \ref{sec:application} is devoted to the application of Hermitian hulls to EAQECCs. The last section contains concluding remarks. 

\section{Preliminaries}\label{sec:pre}
An $\F_{q}$-linear code $C$ of length $n$, dimension $k$, and minimum distance $d$ is  an $[n,k,d]_q$ code. If $d=n-k+1$, then $C$ is {\em maximum distance separable} ({MDS}). 

The {\em Hermitian inner product} of vectors ${\bf{a}}=(a_1,\ldots, a_n)$ and ${\bf{b}}=(b_1, \ldots, b_n)$ over $\F_{q^2}$ is 
$\langle{\bf{a}},{\bf{b}}\rangle_{\rm H}=\sum_{i=1}^n a_i \, b_i^q$. They are \emph{orthogonal} if their inner product is $0$. The Hermitian {\em dual} of a code $C$, denoted by $C^{\perp_{\rm H}}$, is the set of all vectors which are orthogonal to every codeword of $C$. The code $C$ is \emph{Hermitian self-orthogonal} if $C\subseteq C^{\perp_{\rm H}}$. The \emph{Hermitian hull} of $C$ is 
${\rm Hull}_{\rm H}(C):= C \cap C^{\perp_{\rm H}}$. The Euclidean case is defined analogously by $\langle{\bf{a}},{\bf{b}}\rangle_{\rm E} := \sum\limits_{i=1}^n a_i b_i$.

Given a code $C\subseteq \F_{q^2}^n$ and vectors ${\bf a}=(a_1,\ldots,a_n), {\bf b}=(b_1,\ldots,b_n) \in (\F_{q^2}^*)^n$, we define
\begin{align*}
{\bf a}^q &:=(a_1^q,\ldots,a_n^q), &
\frac{1}{{\bf a}}&:= \left(\frac{1}{a_1},\ldots,\frac{1}{a_n}\right),\\
{\bf a} \, {\bf b} &:=({a_1}{b_1},\ldots,{a_n}{b_n}), &
{\bf a} \, C &:=\{{\bf a} {\bf c}: {\bf c}\in C\}.
\end{align*}

We recall notions related to the algebraic functions of one variable to define algebraic geometry (AG) codes and use Stichtenoch's textbook  \cite{Stich} as the reference for terms that we do not have the space to formally define. 

The \emph{function field} of an algebraic curve $\cal X$ over $\F_q$ is a
finite separable extension of $\F_q(x)$, with $x$ being a transcendental element over $\F_q$. We denote by $\F_q({\cal X})$ the function field of ${\cal X}$. Since $\cal X$ is henceforth fixed to be a \emph{rational curve}, we use the notation $\F_q(x)$ instead. A \emph{place} $P$ of $\F_q(x)/\F_q$ is the maximal ideal of the valuation ring ${\cal O}_P$. A \emph{point} on ${\cal X}$ can be identified with the place of the function field $\F_q(x)/\F_q$. A place at infinity is denoted by $O$. We define a \emph{divisor} $G$ on ${\cal X}$ to be a formal sum $\sum\limits_{P\in {\cal X}}n_PP$ with only finitely many nonzeroes $n_P \in \Z$. A divisor $G$ is \emph{rational} if, for any $\sigma \in \text{Gal}(\overline{\F_q}/\F_q)$, we have $G^\sigma=G$.
The \emph{support} of $G$ is the set ${\rm supp}(G):=\{P\in {\cal X} : n_P \not=0\}$. If $G=\sum\limits_{P\in {\cal X}}n_PP$, then the \emph{degree} of $G$ is $\deg(G) := \sum\limits_{P\in {\cal X}}n_P\deg(P)$, where $\deg (P)$ is the size of the orbit of $P$ under the action $\sigma$.
For a nonzero rational function $z$ on the curve $\cal X$, the \emph{principal divisor} of $z$ is $(z) := \sum\limits_{P\in {\cal X}}v_P(z) \, P$ with $v_P$ being the normalized discrete valuation corresponding to the place $P$. For any $z$ in the local ring ${\cal O}_P$ such that $z=u \, t^s$, with $u$ being a unit and $t$ a generator of the maximal ideal of ${\cal O}_P$, we have $v_P(z):=s$. If $Z(z)$ and $N(z)$ denote the respective sets of zeroes and \emph{poles} of $z$, then the \emph{zero} and \emph{pole divisors} of $z$ are, respectively, 
\[
(z)_0:=\sum\limits_{P\in Z(z)}v_{P}(z)P \mbox{ and } 
(z)_\infty:=\sum\limits_{P\in N(z)}-v_{P}(z)P.
\]
Using this notation, the principal divisor $(z)$ can be written as $(z)=(z)_0-(z)_\infty$. It is well-known that for any rational function $z$, the degree of $(z)$ is equal to zero.

For a divisor $G$ on the curve $\cal X$, we define the \emph{Riemann} and \emph{differential spaces} associated with $G$, respectively, as
\begin{align}
{\cal L}(G) &:=\{z\in \F_q({\cal X}) \setminus \{0\} : (z)+G\ge 0\}\cup \{0\} \mbox{ and} \label{eq:RiemSpace}\\
{\Omega}(G) &:=\{\omega\in \Omega \setminus \{0\} : (\omega)-G\ge 0\}\cup \{0\}, \label{eq:DiffSpace}
\end{align}
where $\Omega:=\{z \, dx : z \in \F_q({\cal X})\}$ is the set of \emph{differential forms} on $\cal X$. Both ${\cal L}(G)$ and ${\Omega}(G)$ are finite-dimensional vector spaces. Let $\ell(G)$ denote the dimension of ${\cal L}(G)$. For any differential form $\omega$ on $\cal X$, there exists a unique rational function $z$ on $\cal X$ such that $\omega=z \, dt$, where $t$ is a \emph{separating element}.
The divisor class of a nonzero differential form is called the \emph{canonical divisor}. Any canonical divisor $K$ on a rational curve has degree $-2$.

We are now ready to define two codes. Let $D :=P_{1}+\cdots+P_{n}$, with $P_{i}$ being a place of degree one  for each $1 \le i \le n$. If $G$ is a divisor having a disjoint support with that of $D$, then the AG and differential AG codes with respect to $D$ and $G$ are defined, respectively, by
\begin{align}
C_{\cal L}(D,G) &:=\{(z(P_{1}),\hdots,z(P_{n})):z\in {\cal L}(G)\} \mbox{ and } \label{eq:CDG}\\
C_{\Omega}(D,G) & :=\{({\rm Res}_{P_{1}}(\omega),\hdots,{\rm Res}_{P_{n}}(\omega)):\omega\in {\Omega}(G-D)\}, \label{eq:ODG}
\end{align}
where ${\rm Res}_{P}(\omega)$ denotes the \emph{residue} of $\omega$ at point $P$.

The parameters of a rational AG code $C_{\cal L}(D,G)$ are given in \cite[Theorem 2.2.2, Corollary 2.2.3]{Stich}. 
\begin{enumerate}
\item A rational AG code $C_{\cal L}(D,G)$ has
\begin{equation*}
k=\ell(G)-\ell(G-D)\text { and } d\ge n-\deg (G).
\end{equation*}
\item Moreover, if $ -2 < \deg(G) < n$, then $C_{\cal L}(D,G)$ has
\begin{equation*}
k=\deg (G)+1\text{ and } d\ge n-\deg (G).
\label{eq:distance}
\end{equation*}
\end{enumerate}

The Euclidean dual of $C_{\cal L}(D,G)$ is useful for proving some results related to the Hermitian hulls.
\begin{lemma}{\rm \cite[Theorem 2.2.8, Proposition 2.2.10] {Stich}}\label{lem:dual2} 
If a differential form $\omega $ satisfies
$v_{P_i}(\omega) =-1$ for $1\le i \le n$ and 
${\rm Res}_{P_i}(\omega) =v_i \neq 0$ for $1\le i \le n$, then 
\begin{equation}\label{eq:EuclideanDual}
C_{\cal L}(D,G)^{\perp_{\rm E}} = C_{\Omega}(D,G) = {\bf v} \, C_{\cal L}(D,H)
\end{equation}
for $H=D-G+(\omega)$ and ${\bf v}=(v_1,\ldots,v_n)$. 
\end{lemma}

Given ${\bf a}=(a_1,\ldots,a_n)$ and ${\bf b} = (b_1,\ldots,b_n)$ in $\F_q^n$ with $a_1,\ldots,a_n$ being all distinct and $b_1,\ldots,b_n$ being all nonzeroes, the generalized Reed-Solomon (GRS) code 
\[
{\rm GRS}_{k}({\bf a},{\bf b}):= 
\{b_1 f(a_1),\ldots,b_nf(a_n) : f\in \F_{q}[x], \deg(f) \le k-1\}
\]
is an MDS code. We know from \cite[Proposition 2.3.3]{Stich} that any rational AG code $C_{\cal L}(D,G)$ with $\deg(G)=k-1$ is equivalent to $GRS_k({\bf a},{\bf b})$, with
\begin{align*}
a_i &=x(P_i) \mbox{ and}\\
b_i &= u(P_i)
\mbox{ for } u(x)\in \F_q(x) \mbox{, with }
(u)=(k-1)P_{\infty}-G .
\end{align*}

It is then immediate to confirm, for $0\le j\le k-1$, that the vectors
\[
(u \, x^j(P_1), \ldots,u \, x^j(P_{n}))=(b_1a_1^j,\ldots,b_{n}a_{n}^j)
\]
form a basis of $C_{\cal L}(D,G)$, allowing us to construct, for $C_{\cal L}(D,G)$, a generator matrix 
\[
\mathcal{G}_{\cal L}(D,G):=
\begin{pmatrix}
b_1&b_2&\hdots&b_{n}\\
b_1 \, a_1& b_2 \, a_2& \cdots &b_{n} \, a_{n}\\
\vdots&\vdots&\cdots&\vdots\\
b_1 \, a_1^{k-2}& b_2 \, a_2^{k-2} & \ddots & b_{n} \, a_{n}^{k-2}\\
b_1 \, a_1^{k-1} & b_2 \, a_2^{k-1} & \cdots & b_{n} \, a_{n}^{k-1}\\
\end{pmatrix}.
\]

\section{Main results}\label{sec:main}

Let $\cal X$ be a rational curve and let $\F_{q^2} ({\cal X})$ be a rational function field, that is, $\F_{q^2} ({\cal X}) = \F_{q^2} (x)$. For a fixed $U\subseteq \F_{q^2}$, the set of its affine coordinates is ${\cal P}_U=\{(a,b)\in {\cal X}(\F_{q^2}):a\in U\}$. 
To study one-point AG codes with explicit hull dimensions, we examine the properties of 
$h(x):=\prod\limits_{\alpha\in U}(x-\alpha)$ and its derivative $h'(x)$. For any point $P=(a,b)$, let $x_P : = (x-a)$. From the choice of $U$, the image of $x_P$ at the local ring at $P$ is a uniformizing parameter. Given the differential form $\omega=\frac{dx}{h(x)}$, for any $P=(a,b)\in {\cal P}_U$, we obtain 
\begin{align*}
\omega &=\frac{d \, x_P}{\prod\limits_{\alpha\in U}(x_P+a-\alpha)} 
=\frac{1}{h'(x_P+a)} \, \frac{h'(x_P+a)}{\prod\limits_{\alpha\in U}(x_P+a-\alpha)} \, d \, x_P\\
&=\frac{1}{h'(x_P+a)} \, \sum\limits_{\alpha\in U}\frac{1}{x_P+a-\alpha} \, d \, x_P.
\end{align*}
Hence, $\omega$ has poles of order one at each $P \in {\cal P}_U$ and its residue at $P$ is ${\rm Res}_{P}(\omega)=\frac{1}{h'(a)}$.
 
From hereon, we fix the differential form $\omega=\frac{dx}{h(x)}$, with $h(x)=\pm \prod\limits_{\alpha\in U} (x-\alpha)$, to guarantee that the conditions on $\omega$ in Lemma \ref{lem:dual2} are met. 
If $G= k \, O$ and $H=D-G+(\omega)$, then $H=(n-k-2) \, O$, with $n$ being the cardinality of the set of evaluation points $U$. We now provide a good lower bound on the Hermitian Hull Dimensions of Rational AG Codes.

For ${\bf v}=(v_1,\ldots,v_n) \in \left(\F_{q^2}^*\right)^n$, the generalized algebraic geometry code associated with ${\bf v}$ is
\[
{\bf v} \, C_{\cal L}(D,G):=\{ \left(v_1 \, z(P_{1}),\ldots, v_n \, z(P_{n})\right) : z\in {\cal L}(G)\}.
\]
It is straightforward to confirm that $ C_{\cal L}(D,G)$ and ${\bf v} \, C_{\cal L}(D,G)$ have the same parameters. 
\begin{lemma}\label{lem:dim-hull} 
Let divisors $D=P_1+\cdots+P_n$ and $G=k \, O$ be such that ${\rm supp}(D)\cap {\rm supp}(G)=\emptyset$. Given a differential form $\omega$, let $H:=D-G+(\omega)$. Let $V_1=\{\x^i:0\le i\le k\}$ and $V_2=\{\x^i:0\le i\le n-k-2\}$ be respective bases of ${\cal L}(G)$ and ${\cal L}(H)$.
Let 
\begin{align}
N&=\min_{ 1 \leq i < q^2} \{i : \x^i(P_j)=1 \mbox{ for all } 1\le j\le n\} \mbox{ and}\label{eq:N}\\
L(N)&=\{i \Mod N : \x^i \in V_1^q\cap V_2\}.\label{eq:L}
\end{align}
If there is a vector ${\bf v}=(v_1,\hdots,v_n)\in (\F_{q^2}^*)^n$ such that ${\rm Res}_{P_i}(\omega)=v_i^{q+1}$ for any $1\le i \le n$, then the Hermitian hull dimension of ${\bf v} \, C_{{\cal L}}(D, G)$ is $\ell\ge |L(N)|$.
\end{lemma}
\begin{proof} 
For brevity, let $C$ stand for $C_{{\cal L}}(D,G)$ and $C'$ for $C_{{\cal L}}(D,H)$. Let 
\[
{\rm Res}(\omega)=({\rm Res}_{P_1}(\omega),\hdots,{\rm Res}_{P_n}(\omega)).
\]
After some computation, we get 
\[
(({\bf v} \, C) \cap ({\bf v} \, C)^{\perp_{\rm H}})^q 
= ({\bf v^q} \, C^q) \cap ({\bf v} \, C)^{\perp_{\rm E}} 
=({\bf v}^q \, C^q) \cap 
\Big( \Big({\rm Res}(\omega) \, \frac{1}{{\bf v}} \Big) \, C' \Big) = 
{\bf v}^q \,( C^q \cap  C'),
\]
where the second equality follows from Lemma \ref{lem:dual2}. 
From the last equality and \cite[Proposition 11]{PerPel}, the Hermitian hull dimension of ${\bf v} \, C$ is $\ell \geq |\{i \Mod N: \x^i \in V_1^q\cap V_2\}|$.
\end{proof}

The next theorem gives an explicit formula to compute the Hermitian hull dimension.
\begin{theorem}\label{thm:2new} 
Let $q$ be a prime power and let us assume Lemma \ref{lem:dim-hull}, with $N$ defined as in \eqref{eq:N}. Let $n = n_0 \, q + q_1$, with $1 \leq n_0\le q-1$, $0\le q_1 \le q-1$, and $k=k_0 \, q + q_0$, with $1 \leq k_0 < \lfloor (q_1+n_0 \, q - q_0)/q \rfloor$, $0\le q_0\le q-1$, and $q_1 - q_0 \le 1$.
\begin{enumerate}
\item If $L(N)$ is as in \eqref{eq:L} with $N=q^2-1$, then ${\bf v} \, C_{{\cal L}}(D, G)$ is an $[n,k+1,n-k]_{q^2}$ MDS code whose Hermitian hull has dimension $\ell\ge |L(q^2-1)|$, where
\begin{equation}
|L({q^2-1})|=
\begin{cases}
k_0(n_0-k_0)+q_0+1 \mbox{, if } k_0\le q_1+q-q_0-2 \mbox{ and } q_0\le n_0-k_0-2,\\
(k_0+1)(n_0-k_0) \mbox{, if } k_0\le q_1+q-q_0-2 \mbox{ and } \\
\qquad q_0\ge n_0-k_0-1, k_0(n_0-k_0-1)+q_1+q \mbox{, with } \\
\qquad k_0> q_1+q-q_0-2 \mbox{ and } q_0< n_0-k_0-2,\\
(n_0-k_0-1)(k_0+1)+(q_1+q-q_0-1) \mbox{, if } \\
\qquad k_0> q_1+ q-q_0-2 \mbox{ and }q_0 \ge n_0-k_0-2.\\
\end{cases}
\label{eq:ell-general}
\end{equation}
\item If $N$ is a proper divisor of $q^2-1$, then the Hermitian hull of ${\bf v} \, C_{{\cal L}}(D, G)$ has dimension $\ell \ge |L(N)| \geq |L(q^2-1)|$, with $|L(q^2-1)|$ as given in \eqref{eq:ell-general}.
\end{enumerate}
\end{theorem}
\begin{proof} 
Based on $V_1 =\{\x^i : 0\le i\le k\}$ and 
$V_2 =\{\x^i : 0\le i\le n-k-2\}$, we partition $V_1^q$ and $V_2$ as
\begin{align*}
V_1^q & =\left(\bigcup_{s=0}^{q-1} \left\{\x^{r+qs} : 0\le r\le k_0-1\right\}\right) \bigcup T_1 \mbox{ and}\\  
V_2 &= \left(\bigcup_{r=0}^{q-1} \left\{\x^{r+qs} : 0\le s\le n_0-k_0-2\right\}\right) \bigcup T_2,
\end{align*}
where $T_1 =\{\x^{k_0+qs}:0\le s\le q_0\}$ and 
$T_2 =\left\{\x^{(n_0-k_0-1)q+r}:0\le r\le q_1+q-q_0-2\right\}$. Expressing $r_1=q_1+q-q_0-2$, we write the respective sets of exponents modulo $N$ of $\x$ in $V_1^q$ and $V_2$ as in \eqref{eq:Part1} and \eqref{eq:Part2}.
\begin{multline}\label{eq:Part1}
\left\{
0, 0+q, \ldots, 0+q_0 \, q, \ldots, 0+(q-1) \,q, \quad 1, 1+q , \ldots , 1+q_0 \, q, \ldots, 1+(q-1) \, q , \ldots, \right. \\
\left. k_0-1 , k_0-1 + q,  \ldots,  k_0-1 + q_0 \, q, \ldots, k_0-1 + (q-1) \, q, \quad {k_0}, {k_0+q}, \ldots, k_0 + q_0 \, q \right\}
\end{multline}
and
\begin{multline}\label{eq:Part2}
\left\{
0, 0+q, \ldots, 0+(n_0-k_0-2) \, q,  0+(n_0-k_0-1) \, q, \right. \\
\left. \quad 1, 1+q, \ldots , 1+(n_0-k_0-2) \, q,  1+(n_0-k_0-1) \, q , \ldots, \right. \\
\left. \ldots , r_1 , r_1+q , \ldots , r_1+(n_0-k_0-2) \, q,   r_1+(n_0-k_0-1) \, q, \ldots, \right. \\
\left. \ldots, q-1, q-1+q, \ldots , q-1+ (n_0-k_0-2) \, q
\right\}
\end{multline}

If $N=q^2-1$, then the set of exponents modulo $N$ of $\x$ in the intersection basis $V_1^q\cap V_2$ has one of the following four forms.
\begin{enumerate}
\item If $k_0\le r_1$ and $q_0\le n_0-k_0-2$, then 
\begin{multline}\label{eq:L1}
L(q^2-1) = \left\{ 0, 0+q, \ldots, 0+q_0 \, q, \ldots 0+(n_0-k_0-1) \, q,
\right. \\
\left. \quad 1, 1+q, \ldots, 1+q_0 \, q, \ldots, 1+(n_0-k_0-1) \, q,
\ldots, \right. \\
\left. \ldots,  k_0-1, k_0-1+q, \ldots, k_0-1+q_0 \, q , \ldots, k_0-1+(n_0-k_0-1)\, q, \right. \\
\left. k_0, k_0+q, \ldots, k_0+ q_0 \, q
\right\},
\end{multline}
with $|L(q^2-1)|=k_0(n_0-k_0)+q_0+1$.

\item If $k_0\le r_1$ and $q_0\ge n_0-k_0-1$, then 
\begin{multline}\label{eq:L2}
L(q^2-1)= \left\{ 0, 0+q, \ldots, 1+(n_0-k_0-2) \, q , 0+(n_0-k_0-1) \, q,
\right. \\
\left. \quad 1, 1+q, \ldots, 1+(n_0-k_0-2) \, q, 1+(n_0-k_0-1) \, q, 
\ldots, \right. \\
\left. \ldots k_0-1, k_0-1+q, \ldots, k_0-1+(n_0-k_0-2) \, q, k_0-1+(n_0-k_0-1) \, q, \right. \\
\left. k_0, k_0+q, \ldots, k_0+(n_0-k_0-2) \, q , k_0+(n_0-k_0-1) \,q \right\},
\end{multline}
with $|L(q^2-1)|=(k_0+1)(n_0-k_0)$.

\item If $k_0> r_1$ and $q_0\le n_0-k_0-2$, then
\begin{multline}\label{eq:L3}
L(q^2-1)= \left\{ 0, \ldots, 0+q_0 \, q, \ldots, 0+(n_0-k_0-2) \, q, 0+ (n_0-k_0-1) \, q, \right. \\
\left. \quad 1, \ldots, 1+q_0 \, q, \cdots, 1+ (n_0-k_0-2) \, q, 
1+(n_0-k_0-1) \, q, \ldots, \right. \\
\left. \ldots, r_1, \ldots , r_1+ q_0 \, q, \ldots, r_1+(n_0-k_0-2) \, q, r_1+(n_0-k_0-1) \, q, \ldots, \right. \\
\left. \ldots, {k_0-1}, \ldots, k_0-1+q_0 \, q, \ldots, k_0-1+ (n_0-k_0-2) \, q, \right. \\
\left. k_0, \ldots, k_0+q_0 \, q, \right\},
\end{multline}
with $|L(q^2-1)|=k_0 \, (n_0-k_0-1) + q_1 + q$.

\item If $k_0> r_1$ and $q_0\ge n_0-k_0-1$, then 
\begin{multline}\label{eq:L4}
L(q^2-1)= \left\{ 0, 0+q, \ldots, 0+(n_0-k_0-2) \, q, 0+(n_0-k_0-1) \, q,
\right. \\
\left. \quad 1, 1+q, \ldots, 1+(n_0-k_0-2) \, q , 1+ (n_0-k_0-1) \, q,
\ldots, \right.\\
\left. \ldots, r_1, q-q_0-2+q, \ldots, r_1+(n_0-k_0-2) \, q, 
r_1+(n_0-k_0-1)\, q, \ldots, \right.\\
\left. \ldots,  k_0, k_0+q, \ldots, k_0+(n_0-k_0-2) \, q \right\},
\end{multline}
with $|L(q^2-1)|=(n_0-k_0-1)(k_0+1)+(q_1+q-q_0-1)$.
\end{enumerate}
We use Lemma \ref{lem:dim-hull} to settle the first assertion of the theorem, with $\ell\ge |L(q^2-1)|$, with $|L(q^2-1)|$ as defined in \eqref{eq:ell-general}.

If $N$ is a proper divisor of $q^2-1$, then 
\begin{equation}\label{eq:LinLprime}
\{i\Mod {(q^2-1)}:\x^i\in V_1^q\cap V_2\} \subseteq 
\{i\Mod N: \x^i \in V_1^q\cap V_2\},
\end{equation}
which ensures that the second part follows from Lemma \ref{lem:dim-hull}.
\end{proof}

\begin{remark} 
We state four useful facts.
\begin{enumerate}
\item For any $1 \leq i \leq n$, let $\omega$ be such that ${\rm Res}_{P_i}(\omega)$ is a $(q+1)^{\rm st}$ power in $\F_{q^2}$. Since $|\{i \Mod N: \x^i \in V_1^q\cap V_2\}|$ is easier to compute than ${\rm rank}(G \, G^\dag)$, with $G$ being a generator matrix of the code, Theorem \ref{thm:2new} supplies a better way to determine a good lower bound on the Hermitian hull dimension of ${\bf v} \, C_{{\cal L}}(D, G)$ for any divisor $N$ of $q^2-1$. 

\item Pereira {\it et al.} \cite{PerPel} determined the Hermitian hull dimension $\ell$ of an $[n,k,n-k+1]_{q^2}$ code for $n=q^2$ and ${\rm Res}(\omega)=(1,\ldots,1)$, where $\ell$ can only take two possible values. Our approach generalizes the results. We remove the constraint on the code length $n$, allowing it to take values other than $q^2$, and our ${\rm Res}_{P_i}(\omega)=v_i^{q+1}$ is not limited to $v_i=1$.

\item If a primitive element $\theta \in \F_{q^2}$ is in the set of evaluation points, then $N$ is always $q^2-1$ since, in this case, we have an element of order $q^2-1$. If the set of evaluation points contains $\theta^e$ for some odd integer $e$, then $N=q^2-1$.
\item If all elements in the set of evaluation points have an even exponent, then $N$ is a \emph{proper divisor} of $q^2-1$.
\end{enumerate}
\end{remark}

We continue our investigation into sets of evaluation points for which $|L(N)|$ can be explicitly computed. We start with a simple one from a multiplicative subgroup of $\F_{q^2}^*$ before presenting a construction of a one-point generalized rational AG code whose Hermitian hull dimension is $\ell \geq |L(N)|$.

Let $U=U_{n-1}\cup \{0\}$, with $U_{n-1}=\{\alpha\in \F_{q^2}: \alpha^{n-1}=1 \}$. If $h(x)=\prod\limits_{\alpha\in U} (x-\alpha)$, then $h'(x) = n \, x^{n-1} - 1$. Hence, $h'(\alpha)=n-1$ for any $\alpha\in U_{n-1}$ and $h'(0)=-1$. Thus, for any $\alpha\in U$, there exists a $\beta \in \F_{q^2}$ such that $h'(\alpha)=\beta^{q+1}$.

\begin{cor}\label{cor:3} 
Let $q$ be a prime power and let $n \neq q^2$ be such that $(n-1)$ divides $(q^{2}-1)$. If 
\begin{align*}
n &= n_0 \, q + q_1 \mbox{, with }1 \leq n_0\le q-1, \, 0\le q_1 \le q-1 \mbox{ and}\\
k &= k_0 \, q + q_0 \mbox{, with } 
1 \le k_0 < \lfloor (q_1 + n_0 \, q -q_0)/q \rfloor, \, 0 \le q_0 < q \mbox{, and } q_1-q_0 \le 1,
\end{align*}
then there exists an $[n, k+1, n-k]_{q^2}$ MDS code whose Hermitian hull is of dimension $\ell \geq |L(N)|\ge |L(q^2-1)|$ as shown in Theorem \ref{thm:2new}.
\end{cor}
\begin{proof} 
We keep $\omega=\frac{dx}{h(x)}$, $D=(h(x))_0=P_1+\cdots+P_n$, $G=k \, O$,  $H = D-G+(\omega)$, and $v_i={\rm Res}_{P_i}(\omega)$ for any $1 \leq i \leq n$. If $N=n-1$, then, by Theorem \ref{thm:2new}, ${\bf v} \, C_{{\cal L}}(D, G)$, with ${\bf v}=(v_1,\ldots,v_n)$, has Hermitian hull dimension $ \ell \geq |L(N)|\ge |L(q^2-1)|$.
\end{proof}

For some dimension $k$ and length $n$ such that $(n-1)$ is a divisor of $(q^2-1)$, Hermitian self-orthogonal codes with parameters $[n,k,n-k+1]_{q^2}$ were constructed in \cite{SokQSC}. The Hermitian hull dimension of such a code is always greater than the $|L(N)|$ in \eqref{eq:ell-general}, making it suitable for application in quantum coding. Unfortunately, determining the exact value of the hull dimension is difficult. To see how far the lower bound is, we provide the exact values of the hull dimensions for $q \in \{7,9\}$ in Table \ref{table0}. For $q=7$, $n=25$, let $\theta$ be the standard primitive element of $\F_{q^2}$ in \texttt{MAGMA} \cite{Magma}. We use ${\bf v}=(\theta^{45}, \theta^{47},\hdots, \theta^{47})$ and the set of evaluation points whose elements have even exponents, namely
\begin{multline*}
\left\{0, 1, \theta^2, \theta^4, \theta^6, 3, \theta^{10}, \theta^{12}, \theta^{14}, 2, \theta^{18}, \theta^{20}, \theta^{22}, 6, \right. \\
\left.  \theta^{26}, \theta^{28},\theta^{30}, 4, \theta^{34}, \theta^{36}, \theta^{38}, 5, \theta^{42}, \theta^{44}, \theta^{46} \right\},
\end{multline*}
to get a $[25,11,15]_{49}$ code with Hermitian hull dimension $6$. The matrix $A_1$ in \eqref{eq:A1} forms a generator matrix 
$\begin{pmatrix}
I_{11} & A_1
\end{pmatrix}$.
\begin{figure*}[ht!]
\begin{equation}\label{eq:A1}
A_1=\begin{pmatrix}
\theta^{19}&2&\theta^{9}&\theta^{17}&\theta^{42}&\theta^{26}&3&\theta^{18}&3&\theta^{44}&\theta^{39}&\theta^{3}&\theta^{30}&\theta^{5}\\
\theta^{38}&\theta^{41}&\theta^{29}&\theta^{44}&\theta^{17}&\theta^{21}&\theta^{20}&\theta^{18}&\theta^{42}&\theta^{17}&\theta^{35}&\theta^{36}&\theta^{43}&\theta^{29}\\
\theta^{43}&\theta^{22}&\theta^{21}&6&3&\theta^{36}&\theta^{38}&\theta^{17}&\theta^{43}&\theta^{17}&\theta^{47}&\theta^{34}&\theta^{2}&\theta^{5}\\
\theta^{19}&\theta&6&\theta^{38}&\theta^{10}&\theta&\theta^{27}&\theta^{9}&2&5&\theta^{21}&\theta^{20}&\theta^{22}&\theta^{34}\\
1&\theta^{11}&\theta^{37}&\theta^{27}&\theta^{10}&\theta^{37}&\theta^{26}&4&\theta^{42}&\theta^{47}&\theta^{30}&\theta^{28}&\theta^{42}&5\\
\theta^{6}&\theta^{19}&\theta^{26}&\theta^{19}&\theta^{26}&2&\theta^{41}&\theta^{10}&\theta^{44}&\theta^{4}&2&2&\theta^{29}&\theta^{39}\\
\theta^{5}&\theta^{17}&\theta^{26}&1&\theta^{10}&6&\theta^{12}&\theta^{17}&\theta^{14}&\theta^{46}&\theta^{13}&\theta^{42}&\theta^{9}&\theta^{18}\\
4&3&2&5&\theta^{31}&1&\theta^{12}&\theta^{28}&\theta^{13}&3&\theta^{47}&\theta^{31}&\theta^{27}&\theta^{38}\\
\theta^{4}&\theta^{14}&\theta^{34}&\theta^{9}&\theta^{2}&1&\theta^{15}&\theta^{7}&\theta^{3}&\theta^{34}&\theta^{36}&\theta^{44}&\theta^{43}&\theta^{35}\\
\theta&\theta^{20}&\theta^{26}&\theta^{13}&\theta^{5}&\theta^{5}&\theta&\theta^{44}&2&\theta^{10}&1&\theta^{19}&\theta^{42}&\theta^{37}\\
\theta^{3}&\theta^{39}&\theta^{6}&\theta^{27}&\theta^{31}&\theta^{30}&\theta^{28}&\theta^{4}&\theta^{27}&\theta^{45}&\theta^{46}&\theta^{5}&\theta^{39}&\theta^{10}
\end{pmatrix}.
\end{equation}
\end{figure*}

\begin{sidewaystable}
\caption{Parameters of MDS codes, for $q \in \{7,9\}$ and $N=n-1$, whose Hermitian hull dimensions are computed based on Corollary \ref{cor:3}. For each code, we list the set $L(q^2-1)$, whose cardinality $|L(q^2-1)|$ is given by (\ref{eq:ell-general}), and the set $L(N)=\{i \Mod N: \x^i \in V_1^q\cap V_2\}$ whose cardinality $|L(N)|$ is precisely the Hermitian hull dimension $\ell$ of the code.}
\label{table0}
\renewcommand{\arraystretch}{1.2}
\centering
\begin{tabular}{llllcc}
\toprule
$(q,n_0,k_0,q_0,q_1)$&$L(q^2-1)$&$L(N)$ & $[n,k+1,n-k]_{q^2}$ & $|L(q^2-1)|$ & $|L(N)| = \ell$ \\
\midrule
$( 7, 3, 1, 3, 4 )$& $\{ 0, 1, 7, 8 \}$&$\{ 0, 1, 4, 7, 8, 11 \}$&$[ 25, 11, 15 ]_{49}$&$4$&$6$ \\
$( 7, 3, 1, 4, 4  )$& $\{ 0, 1, 7, 8 \}$&$\{ 0, 1, 4, 5, 7, 8, 11  \}$&$ [25, 12, 14 ]_{49}$&$4$&$7$ \\
$( 9, 4, 1, 4, 5 )$& $\{ 0, 1, 9, 10, 18, 19 \}$&$\{ 0, 1, 5, 9, 10, 14, 18, 19, 23 \}$&$[ 41, 14, 28 ]_{81}$&$6$&$9$ \\
$( 9, 4, 1, 5, 5  )$& $\{ 0, 1, 9, 10, 18, 19 \}$&$\{ 0, 1, 5, 6, 9, 10, 14, 18, 19, 23  \}$&$ [ 41, 15, 27 ]_{81}$&$6$&$10$ 
\\
$( 9, 4, 1, 6, 5  )$& $\{0, 1, 9, 10, 18, 19 \}$&$\{ 0, 1, 5, 6, 9, 10, 14, 15, 18, 19, 23 \}$&$[ 41, 16, 26 ]_{81}$&$6$&$11$ \\
$( 9, 4, 1, 7, 5 )$& $\{ 0, 1, 9, 10, 18, 19 \}$&$\{ 0, 1, 5, 6, 9, 10, 14, 15, 18, 19, 23  \}$&$ [ 41, 17, 25 ]_{81}$&$6$&$11$ \\
$( 9, 4, 1, 8, 5 )$& $\{ 0, 1, 9, 10, 18, 19 \}$&$\{ 0, 1, 5, 6, 9, 10, 14, 15, 18, 19 \}$&$[ 41, 18, 24 ]_{81}$&$6$&$10$ \\
\bottomrule
\end{tabular}
\end{sidewaystable}

Next, we prove the existence of a family of $\F_{q^2}$-linear MDS codes whose Hermitian hulls have dimensions that can be nicely lower-bounded. Lemma \ref{lem:dim-hull} requires the existence of a vector ${\bf v}=(v_1,\ldots,v_n)\in (\F_{q^2}^*)^n$ that satisfies ${\rm Res}_{P_i}(\omega)=v_i^{q+1}$ for any $1\le i\le n$. We construct such a vector from a carefully built set of evaluation points.

Let $\theta$ be a primitive element of $\F_{q^2}$. We label the elements of $\F_q$ by $u_1,\ldots, u_q$. To guarantee $N=q^2-1$, we choose an $\alpha\in \F_{q^2}\setminus \F_q$ such that $\alpha=\theta^e$ for some odd $e$. We write $\alpha_{i,j} = u_i \, \alpha + u_j$ for $1\le i\le n_0$ and $1\le j\le q$. Such an $\alpha$ exists for any $q$. Let $U=\{\alpha_{i,j} : 1\le i\le n_0, 1\le j\le q\}$ and $h(x)=-(\alpha^q-\alpha)^{(1-t)}\prod\limits_{\beta\in U}(x-\beta)$. We know from \cite[Construction 4]{SokQSC} that $h'(\beta) \in \F_q$ for any $\beta\in U$. 

\begin{cor}\label{cor:1} 
Let $q$ be a prime power and let $n_0$ be an integer such that $1\le n_0\le q-1$. If $k= k_0 \, q + q_0$ with $1\le k_0< \lfloor (n_0 \, q-q_0)/q \rfloor$ and $0\le q_0\le q-1$, then there exists an $[n_0 \, q, k+1]_{q^2}$ MDS code whose Hermitian hull dimension is $\ell \ge |L(q^2-1)|$, with $|L(q^2-1)|$ as in \eqref{eq:ell-general}.
\end{cor}

\begin{proof} 
Let $U$ and $h(x)$ be as in the above discussion. Let 
\[
\omega:=\frac{dx}{h(x)}, \quad 
D:=(h(x))_0 =P_1+ \cdots+P_n, \quad G :=k \, O \mbox{, and }~ H:=D-G+(\omega).
\]
We pick ${\bf v}=(v_1,\hdots,v_{n_0q})\in (\F_{q^2}^*)^n$ so that ${\rm Res}_{P_i}(\omega)=v_i^{q+1}$ for any $1\le i\le n$. If $N=q^2-1$, then, by Theorem \ref{thm:2new}, ${\bf v} \, C_{{\cal L}}(D, G)$ has Hermitian hull dimension 
$\ell\ge |L(q^2-1)|$ as claimed.
\end{proof}

We propose another family of MDS codes with an explicit lower bound on their Hermitian hull dimensions. The codes are defined based on sets of evaluation points to have the property specified in Theorem \ref{thm:2new}. Such a set can be built by using \cite[Construction 3]{SokQSC} and \cite[Lemma 3.7]{FangFuLiZhu}.
Let $U_{s}$ be a multiplicative subgroup of $\F_{q^2}^*$ of order $s$, which means that $s$ divides $(q^{2}-1)$. Let $r=\frac{s}{\gcd(s,q+1)}$. Let $\alpha_{1} \, U_{s}, \ldots ,\alpha_{\frac{q-1}{r}-1} \, U_{s} $ be distinct cosets of $\F_{q^2}^*$ which are different from $U_{s}$. For $1 \leq t \leq \frac{q-1}{n_{2}}-1$, we define  
\begin{equation}\label{eq:U}
U : =U_{s} \bigcup \Big(\bigcup^{t}\limits_{j=1} \alpha_{j}U_{s}\Big) \bigcup \{0\}\mbox{, say } U=\{a_1,\ldots,a_{(t+1){s}+1}\},
\end{equation}
and write 
\begin{equation}\label{eq:h}
h(x)=\prod\limits_{\alpha\in U}(x-\alpha).
\end{equation}
By \cite[Construction 3]{SokQSC}, we have 
\[
h'(a_i)=\beta_i^{q+1} \mbox{ for any } 
1 \le i\le (t+1){s}+1 \mbox{ and some } \beta_i \in \F_{q^2}.
\]
Since there exists an $i$ such that $1\le i\le \frac{q-1}{n_{2}}-1$ and $\alpha_{i}=\theta^{e_i}$, with $e_i$ being odd, we have to include $\alpha_i \, U_s$ as the first coset in $U$ to guarantee that $N=q^2-1$.
This leads to the next corollary.

\begin{cor}\label{cor:2} 
Let $q$ be a prime power. Let integers $t$, $s$, and $r$ be such that 
\[
s \mbox{ divides } (q^{2}-1), \quad 
r=\frac{s}{\gcd (s,q+1)} \mbox{,  and } ~
1\le t \le \frac{q-1}{r}-1.
\]
Let $n =(t+1) \, s+1$. If the code length $n$ can be written as $n=n_0 \, q + q_1$ for some $1 \le n_0$, $q_1\le q-1$, and $k=k_0 \, q + q_0$, with $1 \le k_0 < \lfloor (q_1+n_0 \, q - q_0)/q \rfloor$, $0\le q_0\le q-1$, and $q_1-q_0 \le 1$, then there exists an $[n, k+1, n-k]_{q^2}$ MDS code whose Hermitian hull is of dimension $\ell\ge |L(q^2-1)|$, with $|L(q^2-1)|$ as in \eqref{eq:ell-general}.
\end{cor}

\begin{proof} Let $U$ and $h(x)$ be as in \eqref{eq:U} and \eqref{eq:h} respectively. Let 
\[
\omega :=\frac{dx}{h(x)}, \quad 
D :=(h(x))_0=P_1+\cdots+P_n, \quad  
G :=k \, O \mbox{, and }~ H: =D-G+(\omega).
\]
We select vector ${\bf v}=(v_1,\ldots,v_{n})\in (\F_{q^2}^*)^n$ that satisfies ${\rm Res}_{P_i}(\omega)=v_i^{q+1}$ for any $1\le i\le n$. If $N=q^2-1$, then, by Theorem \ref{thm:2new}, the Hermitian hull of ${\bf v} \, C_{{\cal L}}(D, G)$ has dimension $\ell$, which is lower bounded by $|L(q^2-1)|$, with $|L(q^2-1)|$ as given in \eqref{eq:ell-general}.
\end{proof}

An earlier example, based on Corollary \ref{cor:3}, when $q=7$, $n=25$, and $N=24$, exhibits a $[25,11,15]_{49}$ code with Hermitian hull dimension $6$. Keeping $n=25$ and $\theta$, we now select $N=48$, 
\begin{multline*}
{\bf v} = 
\left\{1, \theta^{46}, \theta^{45}, \theta^{43}, \theta^{46}, \theta^{45}, \theta^{43}, \theta^{46}, \theta^{45}, \theta^{43}, \theta^{46}, \theta^{45}, \right. \\
\left.  \theta^{43}, \theta^{46}, \theta^{45}, \theta^{43}, \theta^{46}, \theta^{45}, \theta^{43}, \theta^{46}, \theta^{45}, \theta^{43}, \theta^{46}, \theta^{45}, \theta^{43}\right\},
\end{multline*}
and set of evaluation points 
\begin{multline*}
\left \{ 1, \theta^{43}, \theta, \theta^{44}, \theta^2, 0, \theta^6, \theta^7, 3, \theta^{12}, \theta^{13}, \theta^{14}, \theta^{18},  \right. \\
\left. \theta^{19}, \theta^{20}, 6,\theta^{25}, \theta^{26}, \theta^{30}, \theta^{31}, 4, \theta^{36}, \theta^{37}, \theta^{38}, \theta^{42} \right\}
\end{multline*}
to get, by Corollary \ref{cor:2}, a $[25,11,15]_{49}$ code with Hermitian hull dimension $4$. The matrix $A_2$ in \eqref{eq:A2} forms a generator matrix $\begin{pmatrix}
I_{11} & A_2 \end{pmatrix}$.

\begin{figure*}[ht!]
\begin{equation}\label{eq:A2}
A_2=\begin{pmatrix}
\theta^{39}&\theta^{4}&\theta^{22}&\theta^{29}&\theta^{20}&4&6&6&\theta^{21}&\theta^{19}&\theta^{44}&\theta^{14}&\theta&\theta^{22}\\
\theta^{18}&\theta^{4}&\theta^{23}&\theta^{42}&\theta^{28}&3&\theta^{13}&\theta^{17}&\theta^{36}&\theta^{10}&\theta^{46}&\theta^{28}&\theta^{41}&2\\
\theta^{20}&\theta^{42}&\theta^{45}&\theta^{41}&\theta^{44}&\theta^{27}&\theta^{11}&6&\theta^{29}&\theta^{45}&\theta^{46}&\theta^{41}&\theta^{26}&\theta^{25}\\
4&\theta^{33}&\theta^{39}&\theta^{4}&4&\theta^{47}&\theta^{19}&\theta^{11}&\theta^{13}&\theta^{27}&\theta^{22}&\theta^{2}&\theta^{28}&\theta^{47}\\
\theta^{2}&1&\theta^{17}&\theta^{33}&\theta^{45}&\theta^{10}&\theta^{14}&\theta^{9}&\theta^{22}&\theta^{33}&\theta^{14}&\theta^{6}&\theta^{17}&\theta\\
3&\theta^{4}&\theta^{18}&\theta^{13}&\theta^{13}&\theta^{10}&\theta^{39}&\theta^{3}&\theta^{19}&\theta^{9}&\theta^{47}&\theta^{25}&\theta^{30}&\theta^{27}\\
\theta^{21}&\theta^{11}&\theta^{11}&\theta^{15}&\theta^{42}&\theta^{42}&5&\theta^{29}&\theta^{29}&\theta^{7}&6&\theta^{47}&\theta^{2}&\theta^{41}\\
\theta^{12}&\theta^{35}&\theta^{47}&\theta^{37}&\theta^{13}&6&\theta^{25}&\theta^{46}&\theta^{44}&\theta^{6}&\theta^{26}&\theta^{12}&\theta^{12}&\theta^{37}\\
\theta^{31}&\theta^{47}&\theta^{44}&\theta^{28}&\theta^{2}&\theta^{10}&\theta^{38}&\theta^{47}&\theta^{29}&2&\theta^{5}&\theta^{42}&\theta^{21}&\theta^{7}\\
\theta^{4}&5&\theta^{2}&\theta^{47}&\theta^{15}&\theta^{9}&\theta^{46}&\theta^{34}&\theta^{19}&\theta^{23}&\theta^{37}&\theta^{10}&\theta^{25}&\theta^{38}\\
5&\theta&4&\theta^{42}&\theta^{43}&\theta&6&\theta^{9}&\theta^{5}&\theta^{12}&\theta^{10}&\theta^{29}&\theta^{28}&\theta^{44}
\end{pmatrix}.
\end{equation}
\end{figure*}

\section{Application to EAQECCs}\label{sec:application}
	
A qudit \emph{quantum error-correcting code} (QECC) $\mathcal{Q}$ with parameters $[[n,\kappa,\delta]]_q$ is a $q^{\kappa}$-dimensional subspace of the Hilbert space $(\mathbb{C}^q)^{\otimes n}$, over the complex field $\mathbb{C}$, with (quantum) minimum distance $\delta$. Such a quantum code encodes $\kappa$ logical qudits into $n$ logical qudits and is capable of correcting quantum error operators affecting up to $\lfloor(\delta-1)/2\rfloor$ arbitrary positions in the quantum ensemble. A qudit \emph{entanglement-assisted quantum code} (EAQECC) requires the communicating parties to share $c$ pairs of error-free maximally entangled states ahead of time. An $[[n,\kappa,\delta; c]]_q$ EAQECC encodes $\kappa$ logical qudits into $n$ physical qudits, with the help of $n-\kappa-c$ ancillas and $c$ pairs of maximally entangled qudits. The code can correct up to $\lfloor(\delta-1)/2\rfloor$ quantum errors. An EAQECC with $c=0$ is a QECC.

\begin{lemma}{\rm (\cite[Corollary 2]{WilBru} and \rm\cite[Proposition 3.3]{GJG18}}\label{lem:quantum1}) 
If $C$ is an $[n,k,\delta]_{q^2}$ code, then there exists an $[[n,\kappa, \delta; c]]_q$ EAQECC $\mathcal{Q}$ with
\begin{equation}\label{eq:paramsEA}
c = (n-k) - \dim
\left({\rm Hull}_{\rm H}(C)\right)\mbox{ and}\
\kappa = 2k-n+c .\\
\end{equation}
\end{lemma}

The Singleton-like bound for any $[[n,\kappa, \delta; c]]_q$ code $\mathcal{Q}$ in \cite[Corollary 9]{GHW} states that 
\begin{alignat}{5}
\kappa &\le c+\max\{0,n - 2 \delta + 2\},\label{eq:QMDS_small_distance}\\
\kappa &\le n-\delta+1,\label{eq:QMDS_trivial}\\
\kappa
&\le\frac{(n-\delta+1)(c+2\delta-2-n)}{3\delta-3-n} \mbox{, with }
\delta-1\ge\frac{n}{2}.\label{eq:QMDS_large_distance}
\end{alignat}
When equality holds in one of the bounds \eqref{eq:QMDS_small_distance}-\eqref{eq:QMDS_large_distance}, $\mathcal{Q}$ is {\it maximum distance separable} (MDS). By Lemma \ref{lem:quantum1}, an $[n,n-k,k+1]_{q^2}$ classical MDS with $k<\lfloor n/2\rfloor$ and Hermitian hull dimension $\ell$ gives rise to an $[[n,n-k-\ell,k+1;k-\ell]]_q$ MDS EAQECC. Thus, the higher the hull dimension $\ell$, the smaller the number $c$ of pre-shared entangled states becomes. In particular, if $\ell=k$, then we get an $[[n,n-2k,k+1]]_q$ MDS QECC.

We use a propagation rule from \cite{Luo2022} to derive the parameters of new codes from known ones.

\begin{lemma}\label{lem:qu-propagation}{\rm \cite[Theorem 12]{Luo2022}}
Let $q>2$ be a prime power. If there exists a pure $[[n,\kappa,\delta;c]]_q$ code $\mathcal{Q}$ constructed by Lemma \ref{lem:quantum1}, then there
exists an $[[n, \kappa+i, \delta; c+i]]_q$ code $\mathcal{Q}^{\prime}$ that is pure to $\delta$ for each $i
\in\{1,\ldots,\ell\}$, with $\ell$ being the Hermitian hull dimension of the ${q^2}$-ary code $C$ that corresponds to $\mathcal{Q}$.
\end{lemma}
Lemma \ref{lem:qu-propagation} describes a trade-off between the dimension and the number of pre-shared entangled states for a pure EAQECC that has fixed length and minimum distance. Classical linear codes with large Hermitian hull dimensions give rise to more EAQECCs and, thus, such classical codes are of interest in constructing EAQECCS with \emph{both} small and large numbers of pre-shared entangled qudits.

To construct an EAQECC, we need a corresponding classical code $C$ with an exact Hermitian hull dimension. We can verify that, for $q>2$, an $\ell'$-dimensional Hermitian hull code $[n,k,d]_{q^2}$ gives rise to an $\ell$-dimensional Hermitian hull code $[n,k,d]_{q^2}$ for any $\ell\in \{0,\hdots,\ell'\}$. We explain the assertion here for completeness. We borrow some technique from \cite[Lemma 5 and the discussion in Section 5]{Sok1D1} and \cite[Theorem 6]{Luo2022}. Let the Hermitian hull dimension of $C$ is $\ell'= \ell+r'$ for some nonnegative integer $r'$. Let $G_2$ be a generator matrix of an $[n,k,d]_{q^2}$ code $C_2$ in systematic form. Then there exist matrices $A$ and $B$ such that
\[
G_2 :=
\begin{pmatrix}
I_{\ell'} & O & A\\
O & I_{k-\ell'} & B
\end{pmatrix},
\]
where $O$ and $I$ denote, respectively, the zero and identity matrices. The matrix $\begin{pmatrix} I_{\ell'} & O & A \end{pmatrix}$ generates a Hermitian self-orthogonal $[n,\ell']_{q^2}$ code. Let $\alpha\in \F_{q^2}^*$ be such that $\alpha^{q+1}\not=1$. We transform $G_2$ to
\[
G_2'=
\begin{pmatrix}
{\rm diag}(\underbrace{1,\hdots,1}\limits_{\ell}, \, \underbrace{\alpha,\hdots,\alpha}\limits_{r'})& O &A\\
O &I_{k-\ell'}&B\\
\end{pmatrix},
\]
which generates a code $C_2'$ whose parameters are the same as those of $C_2$. The Hermitian hull dimension of $C_2'$, however, is only $\ell$. Thus, $C_2'$ and $(C_2')^{\perp_{\rm H}}$ have the same hull dimension.

We proceed to construct EAQECCs based on the linear codes from Section \ref{sec:main}.
\begin{theorem}\label{thm:Q0}
Let $q$ be a prime power and let $n_0$ be an integer such that $1\le n_0\le q-1$. Let $k = k_0 \, q + q_0$, with $k_0$ and $q_0$ being integers satisfying $1 \le k_0 < \lfloor (n_0 \, q - q_0)/q \rfloor$ and $0\le q_0 \le q-1$. If $\ell$ is as given in \eqref{eq:ell-general} in Theorem \ref{thm:2new}, then there exist EAQECCs ${\cal Q}_1$ and ${\cal Q}_2$ with parameters 
\begin{align}
& [[n_0 \, q, k+1-\ell,  n_0 \, q-(k+1); n_0 \, q-(k+1)-\ell]]_{q} \label{eq:para-Q1} \mbox{ and}\\
& [[n_0 \, q, n_0 \, q-(k+1)-\ell, (k+1);(k+1)-\ell]]_{q}.
\label{eq:para-Q2}
\end{align}
\end{theorem}

\begin{proof} 
Taking the code $C$ in Corollary \ref{cor:1} and applying Lemma \ref{lem:quantum1} yield the parameters in (\ref{eq:para-Q1}). We derive the parameters in (\ref{eq:para-Q2}) by applying Lemma \ref{lem:quantum1} on $C^{\perp_{\rm H}}$.
\end{proof}

Lemma \ref{lem:quantum1} tells us that the larger the Hermitian hull dimension $\ell$ is the smaller the number of pre-shared entangled states $c$ becomes. For the MDS case, a simple approach to check if we have new EAQECC parameters is to fix the $n$ and compare the minimum distance $\delta$ since there is no known propagation rule that can increase the minimum distance of an EAQECC. Known MDS EAQECCs tend to have low minimum distances when they are derived by the Hermitian construction with GRS codes as the classical ingredients. It was shown in \cite{FangFuLiZhu}, for example, that an $[n,k,n-k+1]_{q^2}$ GRS code leads to an MDS EAQECC of minimum distance $\delta=k+1\le \lfloor \frac{n+q-1}{q+1} \rfloor+1$. 

The method that we are proposing here yields classical codes with large dimensions $k$ and explicit Hermitian hull dimensions $\ell$, leading to EAQECCs with relatively large minimum distance $\delta$. The utility is apparent for constructing EAQECCs with small dimensions $\kappa$, large minimum distances $\delta$, and small number of pre-shared entangled states $c$. For a meaningful comparison based on Lemma \ref{lem:qu-propagation}, we fix $n$ and consider the values of $\delta$ obtained by different constructions to exhibit that our codes lead to new parameters not found in prior literature. Without applying any propagation rule, the best EAQECCs in \cite{PerPel} must have length $n=q^2$. Our Theorem \ref{thm:Q0} yields other lengths $n < q^2$. 

Table \ref{table1} lists the new parameters based on Theorem \ref{thm:Q0} for $q\in \{4,5,7\}$ to illustrate the efficacy of our approach. We compare our lower bound $\ell$ on the Hermitian hull dimension of the ingredient classical codes with the lower bound, denoted by $\ell_{\rm HC}$, that was recently obtained by H.~Chen in \cite[Main Result; see also Thm. 2.2]{Chen2024} when $t=0$ and $k \geq \frac{n}{2}$. Our bound is clearly sharper for the listed input parameters.

\begin{sidewaystable}
\caption{Parameters of qudit MDS (marked with $^*$) and almost MDS EAQECCs from Theorem \ref{thm:Q0} for $q\in \{4,5,7\}$. We denote by $\delta^o$ the minimum distance of the EAQECC upon applying the propagation rule in \cite[Corollary 11]{GHW} to the best-known QECC in \cite{Table}. The lower bound $\ell_{\rm HC}$ on the dimension of the Hermitian hulls of the ingredient classical codes are computed based on \cite[Main Result]{Chen2024}. The superior lower bound $\ell$ is ours.}
\label{table1}
\renewcommand{\arraystretch}{1.1}
\setlength{\tabcolsep}{5pt}
\centering
\begin{tabular}{ccc  lc lc |ccc  lc lc   }
\toprule
$(q,n_0,k_0,q_0)$ & $\ell_{\rm HC}$ & $\ell$ &  EAQECC ${\cal Q}_1$&$
\delta^o$ & EAQECC ${\cal Q}_2$&$
\delta^o$&
$(q,n_0,k_0,q_0)$& $\ell_{\rm HC}$ & $\ell$&  EAQECC ${\cal Q}_1$&$
\delta^o$ & EAQECC ${\cal Q}_2$&$
\delta^o$
\\
\midrule
 $( 4, 3, 1, 0)$ & - & $3$
    &$[[ 12, 2, {\bf 8}; 4 ]]_4$&$6$
    &$[[ 12, 4, {\bf 6}; 2 ]]_4^*$&$5$
&$(7, 5, 2, 0)$& - &$7$
    &$[[ 35, 8, {\bf 21}; 13 ]]_7$&$15$
    &$[[ 35, 13, {\bf 16}; 8 ]]_7^*$&$11$
\\

$( 4, 3, 1, 1)$& - & $4$
    &$[[ 12, 2, {\bf 7}; 2 ]]_4^*$&$6$
    &$[[ 12, 2, {\bf 7}; 2 ]]_4^*$&$6$
&$ (7, 5, 3, 0)$& $5$ &$7$
    &$[[ 35, 15, {\bf 14}; 6 ]]_7^*$&$10$
    &$[[ 35, 6, {\bf 23}; 15 ]]_7$&$16$
\\

$( 4, 3, 1, 3)$ & $1$ & $2$
    &$[[ 12, 6, {\bf 5}; 2 ]]_4^*$&$4$
    &$[[ 12, 2, {\bf 9}; 6 ]]_4$&$7$
&$(7, 5, 1, 1)$ & - &$6$
    &$[[ 35, 3, {\bf 27}; 20 ]]_7$&$19$
    &$[[ 35, 20, {\bf 10}; 3 ]]_7^*$&$7$
\\
$( 5, 3, 1, 0)$ & -
    & $3$
    &$[[ 15, 3, {\bf 10}; 6 ]]_5$&$8$
    &$[[ 15, 6, {\bf 7}; 3 ]]_5^*$&$6$
&$(7, 5, 2, 1)$ & -
    & $8$
    &$[[ 35, 8, {\bf 20}; 11 ]]_7$&$14$
    &$[[ 35, 11, {\bf 17}; 8 ]]_7^*$&$12$
\\
$( 5, 3, 1, 1)$ & - & $4$
    &$[[ 15, 3, {\bf 9}; 4 ]]_5$&$7$
    &$[[ 15, 4, {\bf 8}; 3 ]]_5^*$&$6$
&$( 7, 5, 3, 1)$& $5$ & $8$
    &$[[ 35, 15, {\bf 13}; 4 ]]_7^*$&$9$
    &$[[ 35, 4, {\bf 24}; 15 ]]_7$&$17$
\\

$( 5, 3, 1, 4)$ & $1$ & $2$
    &$[[ 15, 8, {\bf 6}; 3 ]]_5^*$&$5$
    &$[[ 15, 3, {\bf 11}; 8 ]]_5$&$8$
&$( 7, 5, 1, 2)$ & - & $7$
    &$[[ 35, 3, {\bf 26}; 18 ]]_7$&$18$
    &$[[ 35, 18, {\bf 11}; 3 ]]_7^*$&$8$
\\

$( 5, 4, 1, 0)$ & - & $4$
    &$[[ 20, 2, {\bf 15}; 10 ]]_5$&$11$
    &$[[ 20, 10, {\bf 7}; 2 ]]_5^*$&$6$
&$(7, 5, 2, 2)$ & - & $9$
    &$[[ 35, 8, {\bf 19}; 9 ]]_7$&$13$
    &$[[ 35, 9, {\bf 18}; 8 ]]_7^*$&$13$
\\

$( 5, 4, 2, 0)$& $3$ & $5$
    &$[[ 20, 6, {\bf 10}; 4 ]]_5^*$&$8$
    &$[[ 20, 4, {\bf 12}; 6 ]]_5$&$10$
&$( 7, 5, 3, 2)$ & $4$ & $8$
    &$[[ 35, 16, {\bf 12}; 3 ]]_7^*$&$9$
    &$[[ 35, 3, {\bf 25}; 16 ]]_7$&$18$
\\

$ (5, 4, 1, 1)$& - & $5$
    &$[[ 20, 2, {\bf 14}; 8 ]]_5$&$11$
    &$[[ 20, 8, {\bf 8}; 2 ]]_5^*$&$6$
&$(7, 5, 1, 6)$& - & $6$
    &$[[ 35, 8, {\bf 22}; 15 ]]_7$&$15$
    &$[[ 35, 15, {\bf 15}; 8 ]]_7^*$&$10$
\\
$(5, 4, 2, 1)$& $3$ & $6$
    &$[[ 20, 6, {\bf 9}; 2 ]]_5^*$&$7$
    &$[[ 20, 2, {\bf 13}; 6 ]]_5$&$10$
&$(7, 5, 3, 6)$& $3$ & $4$
    &$[[ 35, 24, {\bf 8}; 3 ]]_7^*$&$7$
    &$[[ 35, 3, {\bf 29}; 24 ]]_7$&$20$
\\

$( 5, 4, 1, 4)$& - & $4$
    &$[[ 20, 6, {\bf 11}; 6 ]]_5^*$&$9$
    &$[[ 20, 6, {\bf 11}; 6 ]]_5^*$&$9$
&$ (7, 6, 1, 0)$& - & $6$
    &$[[ 42, 2, {\bf 35}; 28 ]]_7$&$25$
    &$[[ 42, 28, {\bf 9}; 2 ]]_7^*$&$7$
\\

$(5, 4, 2, 4)$ & $2$ & $3$
    &$[[ 20, 12, {\bf 6}; 2 ]]_5^*$&$5$
    &$[[ 20, 2, {\bf 16}; 12 ]]_5$&$12$
&$(7, 6, 2, 0)$ & - & $9$
    &$[[ 42, 6, {\bf 28}; 18 ]]_7$&$20$
    &$[[ 42, 18, {\bf 16}; 6 ]]_7^*$&$11$
\\

$( 7, 3, 1, 0)$ & - & $3$
    &$[[ 21, 5, {\bf 14}; 10 ]]_7$&$11$
    &$[[ 21, 10, {\bf 9}; 5 ]]_7^*$&$7$
&$( 7, 6, 3, 0)$ & $8$ & $10$
    &$[[ 42, 12, {\bf 21}; 10 ]]_7^*$&$14$
    &$[[ 42, 10, {\bf 23}; 12 ]]_7$&$16$
\\

$(7, 3, 1, 1)$ & - & $4$
    &$[[ 21, 5, {\bf 13}; 8 ]]_7$&$10$
    &$[[ 21, 8, {\bf 10}; 5 ]]_7^*$&$8$
&$(7, 6, 4, 0)$ & $7$ & $9$
    &$[[ 42, 20, {\bf 14}; 4 ]]_7^*$&$10$
    &$[[ 42, 4, {\bf 30}; 20 ]]_7$&$21$
\\

$( 7, 3, 1, 2)$ & - & {$4$}
    &$[[ 21, 6, {\bf 12}; 7 ]]_7$&$9$
    &$[[ 21, 7, {\bf 11}; 6 ]]_7^*$&$8$

&$( 7, 6, 1, 1)$& - & $7$
    &$[[ 42, 2, {\bf 34}; 26 ]]_7$&$24$
    &$[[ 42, 26, {\bf 10}; 2 ]]_7^*$&$7$
\\

$(7, 3, 1, 6)$ & $1$ & $2$
    &$[[ 21, 12, {\bf 8}; 5 ]]_7^*$&$7$
    &$[[ 21, 5, {\bf 15}; 12 ]]_7$&$11$
&$( 7, 6, 2, 1)$ & - & $10$
    &$[[ 42, 6, {\bf 27}; 16 ]]_7$&$19$
    &$[[ 42, 16, {\bf 17}; 6 ]]_7^*$&$12$
\\

$( 7, 4, 1, 0)$ & - & $4$
    &$[[ 28, 4, {\bf 21}; 16 ]]_7$&$15$
    &$[[ 28, 16, {\bf 9}; 4 ]]_7^*$&$7$
&$( 7, 6, 3, 1)$ & $8$ & $11$
    &$[[ 42, 12, {\bf 20}; 8 ]]_7^*$&$14$
    &$[[ 42, 8, {\bf 24}; 12 ]]_7$&$17$
\\

$ (7, 4, 2, 0)$ & $3$ & $5$
    &$[[ 28, 10, {\bf 14}; 8 ]]_7^*$&$10$
    &$[[ 28, 8, {\bf 16}; 10 ]]_7$&$12$
&$( 7, 6, 4, 1)$ & $6$ & $10$
    &$[[ 42, 20, {\bf 13}; 2 ]]_7^*$&$10$
    &$[[ 42, 2, {\bf 31}; 20 ]]_7$&$22$
\\
$(7, 4, 1, 1)$ & - & $5$
    &$[[ 28, 4, {\bf 20}; 14 ]]_7$&$14$
    &$[[ 28, 14, {\bf 10}; 4 ]]_7^*$&$8$
&$(7, 6, 1, 2)$ & - & $8$
    &$[[ 42, 2, {\bf 33}; 24 ]]_7$&$23$
    &$[[ 42, 24, {\bf 11}; 2 ]]_7^*$&$8$
\\
$( 7, 4, 2, 1)$ & - & $6$
    &$[[ 28, 10, {\bf 13}; 6 ]]_7^*$&$10$
    &$[[ 28, 6, {\bf 17}; 10 ]]_7$&$13$
& $(7, 6, 2, 2)$ & - & $11$
    &$[[ 42, 6, {\bf 26}; 14 ]]_7$&$18$
    &$[[ 42, 14, {\bf 18}; 6 ]]_7^*$&$12$
\\

$(7, 4, 1, 2)$ & - & $6$
    &$[[ 28, 4, {\bf 19}; 12 ]]_7$&$14$
    &$[[ 28, 12, {\bf 11}; 4 ]]_7^*$&$8$
&$(7, 6, 3, 2)$ & $7$ & $12$
    &$[[ 42, 12, {\bf 19}; 6 ]]_7^*$&$13$
    &$[[ 42, 6, {\bf 25}; 12 ]]_7$&$17$
\\

$(7, 4, 2, 2)$ & $3$ & $6$
    &$[[ 28, 11, {\bf 12}; 5 ]]_7^*$&$9$
    &$[[ 28, 5, {\bf 18}; 11 ]]_7$&$13$
&$(7, 6, 4, 2)$ & $6$ & $9$
    &$[[ 42, 22, {\bf 12}; 2 ]]_7^*$&$9$
    &$[[ 42, 2, {\bf 32}; 22 ]]_7$&$23$
\\

$( 7, 4, 1, 6)$& - & $4$
    &$[[ 28, 10, {\bf 15}; 10 ]]_7^*$&$11$
    &$[[ 28, 10, {\bf 15}; 10 ]]_7^*$&$11$
&$(7, 6, 1, 6)$& - & $8$
    &$[[ 42, 6, {\bf 29}; 20 ]]_7$&$20$
    &$[[ 42, 20, {\bf 15}; 6 ]]_7^*$&$10$
\\

$( 7, 4, 2, 6)$ & $2$ & $3$
    &$[[ 28, 18, {\bf 8}; 4 ]]_7^*$&$7$
    &$[[ 28, 4, {\bf 22}; 18 ]]_7$&$16$
&$( 7, 6, 2, 6)$ & - & $9$
    &$[[ 42, 12, {\bf 22}; 12 ]]_7^*$&$15$
    &$[[ 42, 12, {\bf 22}; 12 ]]_7^*$&$15$
\\

$(7, 5, 1, 0)$ & - & $5$
    &$[[ 35, 3, {\bf 28}; 22 ]]_7$&$20$
    &$[[ 35, 22, {\bf 9}; 3 ]]_7^*$&$7$
&$(7, 6, 4, 6)$ & $4$ & $5$
    &$[[ 42, 30, {\bf 8}; 2 ]]_7^*$&$7$
    &$[[ 42, 2, {\bf 36}; 30 ]]_7$&$25$\\
\bottomrule
\end{tabular}
\end{sidewaystable}

\subsection*{Comparison 1}
\begin{itemize}
\item 
Our EAQECCs have minimum distances that are strictly larger than those in \cite{FangFuLiZhu}. When $q=4$ and $n=12$, for example, the largest possible $\delta$ of any MDS EAQECC in~\cite{FangFuLiZhu} is at most $\lfloor \frac{n+q-1}{q+1} \rfloor+1=4$, if it exists. We see in Table \ref{table1} the new parameters $[[12,6,5;2]]_4$, $[[12,4,6;2]]_4$, and $[[12,2,7;2]]_4$, all with $\delta > 4$. Table \ref{table1} presents our new parameters for $q \in \{5,7\}$ for different lengths. 
\item Assisted by $2$ entangled $4$-dits, our $[[12,2,7;2]]_4$ code can correct $1$ more error than the best-known $[[12,2,5]]_4$ QECC in \cite{Table}, which is not entanglement-assisted. 

In general, the existence of an $[[n,n-2(\delta-1),\delta]]_q$ QECC implies the existence of an $[[n-c,n-2(\delta-1),\delta;c]]_q$ EAQECC by a propagation rule in \cite[Corollary 11]{GHW}. The $[[14,2,6]]_4$ QECC in \cite{Table}, for instance, yields a $[[12,2,6;2]]_4$ EAQECC. The minimum distance of our new $[[12,2,7;2]]_4$ code is strictly larger. Applying the propagation rule to known QECCs in \cite{Table}, Table \ref{table1} makes a meaningful comparison on the minimum distance for fixed $(n,\kappa,c)$. 

Columns with header $\delta^o$ in Table \ref{table1} provide the respective best minimum distances $\delta^o$ of the 
$[[n-c,n-2(\delta^o-1),\delta^o;c]]_q$ codes $\mathcal{Q}_0$ that the propagation rule produces based on the parameters of the corresponding $[[n,n-2(\delta^o-1),\delta^o]]_q$ best-known QECCs in \cite{Table}. Our new $\mathcal{Q}_1$ and $\mathcal{Q}_2$ have the same length and dimension as $\mathcal{Q}_0$ but strictly better minimum distance $\delta > \delta^0$.

To highlight the fact that our codes have comparatively better minimum distances, the values are presented in bold. For example, given $n=12$, $\kappa=2$, and $c=4$, the code derived by the propagation rule from the best known $[[16,2,6]]_4$ QECC has parameters $[[12,2,\delta^0=6;4]]_4$. Ours is the $[[12,2,{\bf 8};4]]_4$ code in the third column. Other entries can be similarly interpreted.
\item For $q=7$, the codes we obtain by Theorem \ref{thm:Q0} have lengths that are not covered by those in \cite{ChenHuangFengChen,FanChenXu,LSEL,WangZhuSun,Sari}. To the best of our knowledge, there are no known propagation rules that can derive the parameters of our EAQECCs from previously known ones. 
\end{itemize}

Putting the parameters and setups in perspective, our approach results in classical codes with arbitrary Hermitian hull dimensions, leading to explicit determination of the resulting parameters of the corresponding EAQECCs. We then have the flexibility to design EAQECCs with large minimum distances while keeping the number of required pairs of entangled states small. Looking at the resulting parameters, one can carefully weigh the trade-offs between using the best-known QECCs in \cite{Table} and utilizing those in Table \ref{table1} if $c$ is small and the gain in $\delta$ is significant.

Applying Lemma \ref{lem:quantum1} to the classical codes of Corollary \ref{cor:2} gives us the next result.

\begin{theorem}\label{thm:Q1} 
Let $q$ be a prime power and let integers $t$, $s$, and $r$ be such that 
\[
s \mbox{ divides } (q^{2}-1), \quad 
r=\frac{s}{\gcd (s,q+1)} \mbox{, and }~ 
 1 \leq t \leq \frac{q-1}{r}-1.
\]
We write $n=(t+1) \, s+1$ as $n = n_0 \, q+ q_1$ for some $1 \leq n_0\le q-1$, $0\le q_1 \leq q-1$. We express $k= k_0 \, q + q_0$, with 
\[
1 \leq k_0 < \lfloor (q_1 + n_0 \, q - q_0) /q \rfloor, \quad 
0 \leq q_0 \leq q-1 \mbox{, and }~ q_1-q_0 \leq 1.
\]
If $\ell$ is as in \eqref{eq:ell-general}, then there exist ${\cal Q}_1$ and ${\cal Q}_2$ with parameters
\begin{align}
& [[n, k+1-\ell, n-(k+1);n-(k+1)-\ell]]_{q} \mbox{ and }\\
& [[n, n-(k+1)-\ell, (k+1);(k+1)-\ell]]_{q}.
\end{align}
\end{theorem}

Applying the propagation rule of Lemma \ref{lem:qu-propagation} does not lead to overlapping parameters. Table \ref{table2} lists the parameters of previously known and new $7$-ary MDS EAQECCs. The new ones are computed based on Theorem \ref{thm:Q1} and Corollary \ref{cor:3}.

\begin{table}
\caption{Previously Known and New MDS EAQECCs with Parameters $[[n,\kappa,\delta;c]]_7$.}
\label{table2}
\renewcommand{\arraystretch}{1.2}
\setlength{\tabcolsep}{5pt}
\centering
\begin{tabular}{lc | lc | lc  lc }
\toprule
$[[n,\kappa,\delta;c]]_7$ & Ref. & 
$[[n,\kappa,\delta;c]]_7$ & Ref. &
$[[n,\kappa,\delta;c]]_7$ & Ref.  
\\
\midrule

$[[24,4,13;4]]_7$ &\cite{ChenHuangFengChen} &
$[[25,15,11;10]]_7$ &\cite{FanChenXu} &
$[[41,19,15;6]]_7$ & \cite{LSEL}  
\\

$[[24,6,12;4]]_7$ &\cite{ChenHuangFengChen} &
$[[25,16,10;9]]_7$ &\cite{FanChenXu} &
$[[41, 20, 14; 5 ]]_7$& Thm. \ref{thm:Q1}
\\

$[[24,8,10;2]]_7$ &\cite{FanChenXu} &
$[[25,17,9;8]]_7$ &\cite{FanChenXu} &
$[[41,23,11;2]]_7$ & \cite{LSEL}  
\\

$[[24,10,9;2]]_7$ &\cite{FanChenXu} &
$[[25,18,8;7]]_7$ &\cite{FanChenXu} &
$[[41,25,10;2]]_7$ & \cite{LSEL}  
\\

$[[24,12,8;2]]_7$ &\cite{FanChenXu} &
$[[33,10,16;8]]_7$ & \cite{LSEL}  &
$[[41,27,9;2]]_7$ & \cite{LSEL}  
\\

$[[25,5,13;4]]_7$ &\cite{WangZhuSun} &
$[[33, 13, 15; 8 ]]_7$& Thm. \ref{thm:Q1}&
$[[41,29,8;2]]_7$ & \cite{LSEL}  
\\

$[[ 25, 6, 13; 5 ]]_7$& Cor. \ref{cor:3} and Table 1 &
$[[ 33, 14, 14; 7 ]]_7$& Thm. \ref{thm:Q1}&
$[[49,12,24;9]]_7$ & \cite{LSEL}  
\\

$[[ 25, 8, 12; 5 ]]_7$& Cor. \ref{cor:3} and Table 1&
$[[ 33, 15, 13; 6 ]]_7$& Thm. \ref{thm:Q1}&
$[[49,14,23;9]]_7$ & \cite{LSEL}  
\\

$[[25,9,11;4]]_7$ &\cite{WangZhuSun} &
$[[33,17,10;2]]_7$ & \cite{LSEL}  &
$[[49,16,22;9]]_7$ & \cite{LSEL}  
\\

$[[25,11,9;2]]_7$ & \cite{LSEL}  &
$[[33,19,9;2]]_7$ & \cite{LSEL}  &
$[[49,19,18;4]]_7$ & \cite{LSEL}  
\\

$[[25,13,8;2]]_7$ & \cite{LSEL} &
$[[33,21,8;2]]_7$ & \cite{LSEL}  &
$[[49,21,17;4]]_7$ & \cite{LSEL} 
\\

$[[25,13,9;4]]_7$ &\cite{WangZhuSun} &
$[[ 41, 12, 21; 11 ]]_7$& Thm. \ref{thm:Q1}&
$[[49,23,16;4]]_7$ & \cite{LSEL} 
\\

$[[25,13,13;12]]_7$ &\cite{FanChenXu} &
$[[ 41,15,17;6]]_7$ & \cite{LSEL}  &
$[[49,25,15;4]]_7$ &\cite{LSEL} 
\\

$[[25,14,12;11]]_7$ &\cite{FanChenXu} &
$[[41,17,16;6]]_7$ & \cite{LSEL}  &
$[[49,26,13;1]]_7$ &\cite{FanChenXu} 
\\
\bottomrule
\end{tabular}
\end{table}

\subsection*{Comparison 2} 
Theorem \ref{thm:Q1} gives us a $[[33,13,15;8]]_7$ MDS EAQECC. There is a known $[[33,10,16;8]]_7$ MDS EAQECC from \cite{LSEL} (see Table \ref{table2}). Both codes share the same $n$ and $c$. Our code has more codewords, the same error-correction capability on $7$ arbitrary positions in the quantum ensemble, but with $1$ less error-detection power.

\section{Concluding Remarks}
We have presented our studies on the Hermitian hulls of one-point generalized rational AG codes ${\bf v} \, C_{{\cal L}}(D, G)$ over $\F_{q^2}$, where ${\bf v}=(v_1,\hdots,v_n)\in \left(\F_{q^2}^*\right)^n$, $D=P_1+\cdots+P_n$, $G=k \, O$, $n=n_0 \, q+ q_1$, and $k=k_0 \, q+ q_0$, with specific constraints on $n_0$, $k_0$, $q_0$, and $q_1$. An excellent lower bound on the hull dimensions can be explicitly computed upon careful selection of the corresponding sets of evaluation points. 

Our approach leads to MDS linear codes with designed hull dimensions, resulting in two new families of EAQECCs with excellent parameters. In terms of the classical codes that we use as ingredients to derive the parameters of the corresponding EAQECCs, we have established an excellent lower bound $\ell$ on the dimensions of their respective Hermitian hulls. This bound, illustrated in Table \ref{table1}, is sharper than the recently published bound of Chen from \cite{Chen2024}.

Two open directions emerge from our studies. One can explore if new families of good EAQECCs can be built from other sets of evaluation points on rational curves and with different constraints from those in Theorem \ref{thm:2new}. Another option is to utilize more general algebraic curves.


\end{document}